\documentclass[11pt,a4paper]{amsart}
\usepackage{fullpage}

\usepackage[utf8]{inputenc}
\usepackage[T1]{fontenc}
\usepackage[english]{babel}

\usepackage{amsthm,aliascnt}
\usepackage{amsmath,hyperref}
\usepackage{nicefrac}
\usepackage{nth}
\usepackage{float}

\usepackage{varwidth,bussproofs}
\newenvironment{mathprooftree}
  {\varwidth{.9\textwidth}\centering\leavevmode}
  {\DisplayProof\endvarwidth}

\usepackage{caption, subcaption}
\usepackage{tikz,pgfplots}
\usetikzlibrary{shapes.geometric}
\usetikzlibrary{positioning}

\newcommand{\mynewtheorem}[2]{
  \newaliascnt{#1}{dummy}
  \newtheorem{#1}[#1]{#2}
  \aliascntresetthe{#1}
  \expandafter\def\csname #1autorefname\endcsname{#2}
}

\theoremstyle{definition}
\mynewtheorem{theorem}{Theorem}
  \mynewtheorem{prop}{Proposition}
  \mynewtheorem{cor}{Corollary}
  \mynewtheorem{lemma}{Lemma}

\theoremstyle{definition}
  \mynewtheorem{definition}{Definition}
  \mynewtheorem{example}{Example}

\newcommand{\nats}{\mathbb{N}}
\newcommand{\C}{\mathbb{C}}
\newcommand{\Z}{\mathbb{Z}}
\newcommand{\set}[1]{\{#1\}}

\DeclareMathOperator{\cat}{\mathsf{Cat}}
\newcommand{\catalan}[1]{{\cat}_{#1}}

\renewcommand{\S}{\mathbf{S}}
\newcommand{\K}{\mathbf{K}}

\newcommand{\rg}[1]{\mathcal{R}_{#1}}
\newcommand{\rgf}[1]{R_{#1}(z)}
\newcommand{\domsing}[1]{\zeta_{#1}}

\newcommand{\B}{\mathcal{B}}

\newcommand{\T}[1]{T_{#1}}
\newcommand{\comb}[1]{\mathbf{#1}}

\newcommand{\density}[2]{\mu \left( \frac{#1}{#2}\right)}
\newcommand{\densityplus}[2]{\mu^{+} \left( \frac{#1}{#2} \right) }
\newcommand{\densityminus}[2]{\mu^{-} \left( \frac{#1}{#2} \right) }
\newcommand{\class}[1]{{\mathcal C}_{#1}}

\newcommand{\WN}{\mathcal{WN}}

\begin{document}
\title{On the likelihood of normalisation in combinatory logic}
\author{Maciej Bendkowski \and Katarzyna Grygiel \and Marek Zaionc}
\address{
  Theoretical Computer Science Department\\
  Faculty of Mathematics and Computer Science\\
  Jagiellonian University\\
  ul. Prof. {\L}ojasiewicza 6, 30-348 Krak\'ow, Poland}
\email{\{bendkowski,grygiel,zaionc\}@tcs.uj.edu.pl}
\thanks{This work was partially supported within the grant 2013/11/B/ST6/00975 founded by the Polish National Science Center. Maciej Bendkowski was supported with the computational grant plgmbendkow2015a by ACC Cyfronet AGH in the PL-Grid NGI project}
\keywords{Combinatory logic, analytic combinatorics, normalisation}
\date{\today}

\begin{abstract}
We present a quantitative basis-independent analysis of combinatory logic.
Using a general argument regarding plane binary trees with labelled leaves, we generalise the results of David et al. (see~\cite{dgkrtz}) and Bendkowski et al. (see~\cite{Bendkowski2015}) to all Turing-complete combinator bases proving, inter alia, that asymptotically almost no combinator is strongly normalising nor typeable. We exploit the structure of recently discovered normal-order reduction grammars (see~\cite{Bendkowski16}) showing that for each positive $n$, the set of $\S \K$-combinators reducing in $n$ normal-order reduction steps has positive asymptotic density in the set of all combinators. Our approach is constructive, allowing us to systematically find new asymptotically significant fractions of normalising combinators. We show that the density of normalising combinators cannot be less than $34\%$, improving the previously best lower bound of approximately $3\%$ (see~\cite{Bendkowski2015}). Finally, we present some super-computer experimental results, conjecturing that the density of normalising combinators is close to $85\%$.
\end{abstract}
\maketitle
\vspace{-8mm}
\section{Introduction}\label{sec:introduction}
Quantitative investigations in logic, where properties and behaviour of typical objects are studied, form a rich and well-established branch of mathematics on the border of logic, combinatorics and theoretical computer science. From a combinatorial point of view, logical formulae are objects with finite representations and, sometimes several, naturally associated notions of size. In cases when the assumed size notion imposes finitely many objects of any size, we can consider uniformly random formulae within such a family of objects. Analysing the asymptotic behaviour of the probability that a uniformly random object of size $n$ satisfies a certain property~$P$ as $n$ tends to infinity, yields the notion of asymptotic density of $P$ and, in consequence, leads to investigations how certain natural properties, such as satisfiability etc., behave in the case of typical formulae.

There is a long history of using this kind of asymptotic approach applied to logic (see, e.g.~\cite{mtz00, kos-zaionc03, FGGZ07, GK2012}) and recently its computational aspects. In~\cite{dgkrtz}, David et al. initiated quantitative investigations in lambda calculus and combinatory logic. Considering the set of closed $\lambda$-terms in a canonical representation where variables do not contribute to the term size, David et al. showed that typical $\lambda$-terms are strongly normalising whereas in the case of $\S \K$-combinators the situation is precisely opposite -- asymptotically almost no $\S \K$-combinator is strongly normalising. Somewhat contrary to their result, Bendkowski et al. in~\cite{Bendkowski16} considered a different representation of $\lambda$-terms with de~Bruijn indices, showing that similarly to the case of $\S \K$-combinators, asymptotically almost no $\lambda$-term is strongly normalising. 

Despite many efforts, the associated counting problem for closed $\lambda$-terms is still one of the remaining major open problems. Throughout the years, different variants of lambda calculus have been considered. In~\cite{BGGJ2013}, Bodini et al. studied the enumeration of BCI $\lambda$-terms. John Tromp in~\cite{T2014}, as well as Grygiel and Lescanne in~\cite{JFP:10090684}, considered the counting problem in the so called binary lambda calculus. Recently, in~\cite{gittenberger_et_al:LIPIcs:2016:5741} Gittenberger and Gołębiewski considered $\lambda$-terms in the de~Bruijn notation with various size notions, giving tight lower and upper asymptotic bounds on the number of closed $\lambda$-terms. Due to the lack of bound variables, combinatory logic circumvents the intrinsic issues present in lambda calculus, serving as a minimalistic formalism capable of expressing all lambda-definable functions. Its simple syntax yields a natural representation using plane binary trees with labelled leaves, which greatly facilitates the analysis of asymptotic properties.
 
The paper is organised as follows. In~\autoref{sec:combinatory-logic} we recall preliminary definitions and notation regarding combinatory logic. In the following sections~\ref{sec:combinatorial-classes} and~\ref{sec:generating-functions-and-analytic-tools}, we state basic notions of combinatorial classes and generating functions, listing main tools used to study asymptotic properties of combinators, in particular, the analytic methods of singularity analysis. In~\autoref{sec:basis-independent-results} we study the class of plane binary trees with labelled leaves, deriving basis-independent combinatory logic results as easy corollaries. In~\autoref{sec:sk-combinators} we focus on the $\S \K$-basis proving, inter alia, that for each positive $n$, the set of combinators reducing in $n$ normal-order reduction steps has positive asymptotic density. Finally, in~\autoref{sec:experimental-results} we discuss some experimental results.

\section{Combinatory logic}\label{sec:combinatory-logic}
Let $\B$ be a finite \emph{basis} of primitive combinators. The set $\class{\B}$ of \emph{$\B$-combinators} is defined inductively as follows. If $\comb{X} \in \B$, then $\comb{X} \in \class{\B}$. If $N, M \in \class{\B}$, then $(N M) \in \class{\B}$. In the latter case, we say that $(N M)$ is an \emph{application} of $N$ to $M$. If the underlying basis is clear from the context, we simply write \emph{combinators} instead of $\B$-combinators. Following standard notational conventions (see, e.g.~\cite{BAR84}), we omit outermost parentheses and drop parentheses from left-associated combinators, e.g.~instead of writing $((M N) (P Q))$ we write $M N (P Q)$.

Each primitive combinator $\comb{X} \in \B$ contributes a \emph{reduction rule} of the form
\begin{equation*}
\comb{X} N_1 \ldots N_m \to M
\end{equation*}
where $m \geq 1$, $M \in \class{\set{N_1, \ldots, N_m}}$ and $N_1 \ldots N_m$ are arbitrary combinators. In other words, $\comb{X} N_1 \ldots N_m$ \emph{reduces} to some determined combinator $M$ built from $N_1,\ldots,N_m$ and term application.
The reduction relation $\to$ is then extended onto all combinators such that if $P \to Q$, then $PR \to QR$ and $RP \to RQ$ for each combinator $R$.

Let $P$ be a combinator. If there exists no combinator $Q$ such that $P \to Q$, then $P$ is said to be in \emph{normal form}. If there exists a finite sequence of combinators $P_0,P_1,\ldots,P_k$ such that $P = P_0 \to P_1 \to \ldots \to P_k$ and $P_k$ is in normal form, then $P$ is \emph{weakly normalising}, or simply \emph{normalising}. If there does not exist an infinite sequence of combinators $P_0,P_1,\ldots$ such that $P=P_0 \to P_1 \to \ldots$, then we say that $P$ is \emph{strongly normalising}. Naturally, strong normalisation implies weak normalisation.

A major part of combinatory logic is devoted to its simple type theory corresponding to the implicational fragment of minimal logic (see, e.g.~\cite{BAR84}) and  hence also to simply typed $\lambda$-calculus (see \cite{Hindley}). In its most common basis $\B = \set{\comb{S},\comb{K}}$, the type-assignment deduction system $TA_{\comb{S} \comb{K}}$ is given by the following two axiom schemes for $\comb{S}$ and $\comb{K}$, with a single \emph{modus ponens} inference rule.
\begin{equation}\tag{Axiom $\comb{S}$}
\begin{mathprooftree}
\AxiomC{}
\UnaryInf$\comb{S} : (\alpha \to \beta \to \gamma) \to (\alpha \to \beta) \to \alpha \to \gamma \fCenter$
\end{mathprooftree}
\end{equation}
\begin{equation}\tag{Axiom $\comb{K}$}
\begin{mathprooftree}
\AxiomC{}
\UnaryInf$\comb{K} : \alpha \to \beta \to \alpha\fCenter$
\end{mathprooftree}
\end{equation}
\begin{equation}\tag{Modus ponens}
\begin{mathprooftree}
\AxiomC{$N : \alpha \to \beta$}
\AxiomC{$M : \alpha$}
\BinaryInf$(N M) : \beta\fCenter$
\end{mathprooftree}
\end{equation}
The primitive combinators $\comb{S}$ and $\comb{K}$ can be assigned any types fitting to their axiom schemes. On the other hand, an application $(N M)$ can be assigned a type $\beta$, denoted $(N M) : \beta$, if and only if $N : \alpha \to \beta$, whereas $M : \alpha$. If $P : \alpha$ for some type $\alpha$, we say that $P$ is \emph{typeable}. Naturally, not every combinator is typeable, e.g.~$\omega := \comb{S} \comb{I} \comb{I}$, where $\comb{I} := \comb{S} \comb{K} \comb{K}$. Though not all strongly normalising combinators are typeable, the converse implication holds, i.e.~each typeable combinator is strongly normalising.

In this paper, we focus mostly on bases which are capable of expressing all computable functions. A sufficient and necessary condition for a basis $\B$ to be \emph{Turing-complete} is to span $\B$-combinators extensionally-equivalent to $\comb{S}$ and $\comb{K}$. In such a case, we make the natural assumption that $\B$ defines a set of axiom schemes, one for each primitive combinator, in such a way that when enriched with the modus ponens inference rule, $TA_{\B}$ constitutes a sound typing system -- if $N$ is extensionally-equivalent to either $\comb{S}$ or $\comb{K}$, then the types of $N$ in $TA_{\B}$ are the same as the types of, respectively, $\comb{S}$ or $\comb{K}$ in $TA_{\comb{S} \comb{K}}$. Throughout the paper, we use the over line notation $\overline{\comb{X}}$ to denote arbitrary terms extensionally-equivalent to $\comb{X}$.

We refer the reader to~\cite{curry-feys} or~\cite{BAR84} for a more detailed exposition of combinatory logic.

\section{Combinatorial classes}\label{sec:combinatorial-classes}
Let $B$ be a countable set of objects with an associated \emph{size} function $f \colon B \to \nats$. If for each $n \in \nats$ the set of $B$'s objects of size $n$ is finite, then $B$ together with $f$ forms a \emph{combinatorial class} (see, e.g.~\cite{FlajoletSedgewick2009}). In such a case, we can associate a \emph{counting sequence} $\set{b_n}_{n \in \nats}$ of natural numbers $b_n$ capturing the number of objects in $B$ of size $n$. Naturally, if $A \subseteq B$, then $A$ is a combinatorial class as well with $a_n \leq b_n$ for each $n \in \nats$. Assuming that $b_n > 0$ for each $n\in \nats$, we can then define the \emph{asymptotic density} $\density{A}{B}$ of $A$ in $B$ as

\begin{equation*}
\density{A}{B} = \lim_{n\to\infty} \frac{a_n}{b_n}.
\end{equation*}
Note that if it exists, we can interpret $\density{A}{B}$ as the asymptotic probability of finding an object of $A$ in the class of objects $B$. In other words, the likelihood that $A$ represents `typical' objects in $B$.
Unfortunately, sometimes we do not know whether the asymptotic density of $A$ in $B$ exists, however we can use the \emph{lower} and \emph{upper limits} defined as

\begin{equation*}
\densityminus{A}{B} = \liminf_{n\to\infty} \frac{a_n}{b_n} \quad \text{and}  \quad \densityplus{A}{B} = \limsup_{n\to\infty} \frac{a_n}{b_n}.
\end{equation*}
As $0 \leq \frac{a_n}{b_n} \leq 1$, these two numbers are well defined for any set $A$, even when the limiting ratio $\density{A}{B}$ is not known to exist.
Henceforth, given a combinatorial class $A$, we use $a_n$ to denote the number of objects in $A$ of size $n$.

\section{Generating functions and analytic tools}\label{sec:generating-functions-and-analytic-tools}
Let $A$ be a combinatorial class. The formal power series $A(z) = \sum_{n \geq 0} a_n z^n$ with $A$'s counting sequence $\set{a_n}_{n\in \nats}$ as coefficients, is called the \emph{ordinary generating function} of $A$. Using the powerful theory of \emph{Analytic Combinatorics} developed by Flajolet and Sedgewick~\cite{FlajoletSedgewick2009}, many questions concerning the asymptotic behaviour of $\set{a_n}_{n \in \nats}$ can be efficiently resolved by analysing the behaviour of $A(z)$ viewed as an analytic function in some neighbourhood around the complex plane origin. This is the approach we take to study the asymptotic fractions of interesting combinatory logic terms.

Throughout the paper we use $A(z)$ to denote the ordinary generating function associated with the combinatorial class $A$. We write $[z^n]A(z)$ to denote the coefficient standing by $z^n$ in the Taylor series expansion of $A(z)$ around $z = 0$. We say that two sequences $\set{a_n}_{n \in \nats}$ and $\set{b_n}_{n \in \nats}$ are \emph{asymptotically equivalent} if $\lim_{n \to \infty} \frac{a_n}{b_n} = 1$. In such a case we write $a_n \sim b_n$.

\subsection{Main tools}\label{subsec:main-tools}
In our endeavour to study the asymptotic behaviour of `typical' classes of combinatory logic terms, we use the method of \emph{singularity analysis}~\cite{FlajoletSedgewick2009}. Starting with a particular class $A$ of combinators, we find its corresponding generating function $A(z)$. The location of $A(z)$'s dominant singularities determines the exponential growth rate of $\set{a_n}_{n \in \nats}$ as dictated by the following theorem.

\begin{theorem}[Exponential Growth Formula, see {\cite[Theorem IV.7]{FlajoletSedgewick2009}}]\label{thm:exponential-growth-formula}
			If $A(z)$ is analytic at $0$ and $R$ is the modulus of a singularity nearest to the origin in the sense that
			\[ R = \sup \{ r \geq 0 ~:~ A(z) \text{ is analytic in } |z| < r \} ,\]
			then the coefficient $a_n = [z^n] A(z)$ satisfies
			\[ a_n = R^{-n} \theta(n) \quad \text{with} \quad \limsup |\theta(n)|^{\frac{1}{n}} = 1 .\]
		\end{theorem}
In the case of analytic functions derived from combinatorial classes, the location of dominant singularities is significantly simplified as it suffices to look for singularities on the real line.

\begin{theorem}[Pringsheim, see {\cite[Theorem~IV.6]{FlajoletSedgewick2009}}]
		\label{th:pringsheim}
			If $A(z)$ is representable at the origin by a series expansion that has non-negative
			 coefficients and radius of convergence $R$, then the point $z = R$ is a singularity of $A(z)$.
		\end{theorem}

The sub-exponential factors determining the asymptotic growth rate of $\set{a_n}_{n \in \nats}$ can be then further established using, in our case, the standard function scale for algebraic singularities of square-root type. 

\begin{theorem}[Standard function scale, see~{\cite[Theorem~VI.1]{FlajoletSedgewick2009}}]\label{th:standard-func-scale}
Let $\alpha \in \C \setminus \Z_{\leq 0}$. Then $f(z) = {(1 - z)}^{-\alpha}$ admits for large $n$ a complete asymptotic expansion in form of
\begin{equation*}
[z^n]f(z) = \frac{n^{\alpha-1}}{\Gamma(\alpha)} \left( 1 + \frac{\alpha(\alpha-1)}{2n} + \frac{\alpha(\alpha-1)(\alpha-2)(3\alpha-1)}{24n^2} + O(\frac{1}{n^3}) \right)
\end{equation*}
where $\Gamma$ is the Euler Gamma function.
\end{theorem}

\begin{theorem}[Newton-Puiseux, see~{\cite[Theorem~VII.7]{FlajoletSedgewick2009}}]
Let $f(z)$ be a branch of an algebraic function $P(z, f(z)) = 0$. Then in a circular neighbourhood of a singularity $\zeta$ slit along a ray emanating from $\zeta$, $f(z)$ admits a fractional series expansion that is locally convergent and of the form
\begin{equation*}
f(z) = \sum_{k \geq k_0} c_k {\left( z - \zeta \right)}^{\nicefrac{k}{\kappa}}
\end{equation*}
where $k_0 \in \Z$ and $\kappa \geq 1$.
\end{theorem}

Finally, combining the scaling rule for Taylor expansions and the Newton-Puiseux expansion of algebraic functions with unique dominating singularities, we obtain the following corollary theorem.

\begin{theorem}[Algebraic singularity analysis]\label{th:singularity-analysis}
Let $f(z) = {(1 - \zeta^{-1} z)}^{\nicefrac{1}{2}} g(z) + h(z)$ be an algebraic function, analytic at $0$, having a unique dominant singularity $z = \zeta$. Assume that $g(z)$ and $h(z)$ are analytic in the disk $|z| < \zeta + \eta$ for some $\eta > 0$. Then the coefficient $[z^n]f(z)$ satisfies the following asymptotic approximation
\begin{equation*}
[z^n]f(z) \sim \zeta^{-n} \frac{\overline{C} n^{-\nicefrac{3}{2}}}{\Gamma(-\frac{1}{2})}
\end{equation*}
where $\overline{C}$ is the coefficient standing by $\sqrt{1 - \zeta^{-1} z}$ in the Newton-Puiseux expansion of $f(z)$, i.e.~$g(\zeta)$.
\end{theorem}

\begin{proof}
See~\cite{FlajoletSedgewick2009}, Theorem VII.8.
\end{proof}

In order to simplify the reasoning about the type and location of singularities of generating functions given without explicit closed-form solutions, we use the following technical lemma guaranteeing certain natural closure properties of analytic functions with a single square-root type dominating singularity.

\begin{lemma}\label{lem:sqrt-commutative-ring}
Let $\Omega$ be the open disk $|z| < \zeta + \eta$ for some $0 < \zeta < 1$ and $\eta > 0$. Let $F$ denote the set of functions $f \colon (0,\zeta) \to \C$ in form of $f(z) = \sqrt{1-\zeta^{-1} z}\, P(z) + Q(z)$ for arbitrary $P(z)$ and $Q(z)$ analytic in $\Omega \setminus \set{0}$. Then $F$ with natural function addition and multiplication forms a commutative ring.
\end{lemma}

\begin{proof}
Note that it suffices to show that $F$ is closed under addition and multiplication, as the commutative ring laws are clearly preserved. Let $(U, +, \times)$ be the commutative ring of functions analytic in $\Omega \setminus \set{0}$. Consider arbitrary $f,g \in F$ given by $f(z) = \sqrt{1-\zeta^{-1} z}\, P_f(z) + Q_f(z)$ and $g(z) = \sqrt{1-\zeta^{-1} z}\, P_g(z) + Q_g(z)$. 

Let us start with $f(z) + g(z)$. Note that
\begin{equation*}
 f(z) + g(z) = \sqrt{1-\zeta^{-1} z}\, \big(P_f(z) + P_g(z)\big) + Q_f(z) + Q_g(z). 
 \end{equation*}

Clearly, $f(z) + g(z) = \sqrt{1-\zeta^{-1} z}\, \overline{P}(z) + \overline{Q}(z)$ where both $\overline{P}(z) \in U$ and $\overline{Q}(z) \in U$. Hence, $f(z) + g(z) \in F$.

Now, let us consider $f(z)\cdot g(z)$. By rewriting, we obtain

\begin{eqnarray}\label{eq:sqrt-mult-eq}
f(z)\cdot g(z) &=& \sqrt{1-\zeta^{-1} z} \big(P_g(z) Q_f(z)+P_f(z) Q_g(z)\big)\\
\nonumber && +(1-\zeta^{-1}z)P_f(z) P_g(z) +Q_f(z) Q_g(z).
\end{eqnarray}
Clearly, $f(z)\cdot g(z) \in F$.
\end{proof}

\subsection{Removable singularities}\label{subsec:removable-singularities}
In order to apply~\autoref{th:singularity-analysis} to the analysis of a generating function $A(z)$, we have to guarantee that $A(z)$ is analytic at $z = 0$. If it is not the case, yet $A(z)$ has a \emph{removable pole singularity} at $z = 0$, we can consider its analytic extension $\widetilde{A}(z)$, instead of $A(z)$. The following theorem due to Bernhard Riemann provides a sufficient and necessary condition to determine whether $A(z)$'s pole singularity at $z = 0$ can be removed.

\begin{theorem}[Riemann's Removable Singularities Theorem, see e.g.~\cite{Krantz1999}]\label{th:riemann-rem-sing}
Let $f$ be analytic on the punctured disk $\Omega \setminus \set{z_0}$ of the complex plane. Then $f$ has an analytic extension on $\Omega$ if and only if $\lim_{z \to z_0} (z - z_0) f(z) = 0$. 
\end{theorem}
As a direct consequence, we obtain the following technical lemma.
		
\begin{lemma}\label{lem:sing-powers}
Let $\Omega \setminus \set{z_0}$ be a punctured disk on the complex plane. Suppose that $f$ is analytic on $\Omega \setminus \set{z_0}$ and has an analytic continuation at $z = z_0$. Then for each $n \geq 2$, the function ${f(z)}^n$ has an analytic extension on $\Omega$.
\end{lemma}

\section{Basis-independent results}\label{sec:basis-independent-results}
In this section we are interested in universal basis-independent asymptotic properties of combinatory logic. We prove certain general results about labelled plane binary trees, deriving the combinatory logic results as immediate corollaries.

\begin{definition}
Suppose that $L$ is a finite set of $d$ distinct \emph{labels}. Then, the set of $L$-trees consists of plane binary trees where each leaf has a corresponding label in the set $L$. We use $\T{L}$ to denote the set of $L$-trees.
\end{definition}

Let us notice that the asymptotic growth rate of $L$-trees greatly depends on the asymptotic approximation of Catalan numbers $\catalan{n}$ counting the number of plane binary trees with $n$ inner nodes. It is well known that
\begin{equation*}
\catalan{n} = \frac{1}{n+1}{2n \choose n} \qquad \text{and} \qquad \catalan{n} = 4^n \frac{n^{-\nicefrac{3}{2}}}{\sqrt{\pi}}.
\end{equation*}

\begin{prop}
Let $\T{L}$ be the set of $L$-trees over a set $L$ of size $d$. Suppose that $|\cdot| \colon \T{L} \to \nats$ is a function assigning each $L$-tree $t$ the number of binary nodes in $t$. Then $(\T{L}, |\cdot|)$ forms a combinatorial class.
\end{prop}

\begin{proof}
Let us start with noticing that $t \in \T{L}$ has $|t| + 1$ leaves. Fix $n \in \nats$. The number of plane binary trees with $n$ inner nodes is counted by the $n$th Catalan number $\catalan{n}$. Taking into account all possible $L$-labellings of $n+1$ leaves and using the closed-form expression for $\catalan{n}$, we derive the following formula counting the number $\T{L,n}$ of $L$-terms of size $n$.
\begin{equation*}
\T{L,n} = d^{n+1} \catalan{n} = \frac{d^{n+1}}{n+1}\binom{2n}{n}.
\end{equation*}
\end{proof}

\subsection{Counting $L$-trees containing fixed $L$-trees as subtrees}\label{subsec:counting-l-trees}
Suppose that $t \in \T{L}$. Let $\overline{\T{L}}(z)$ denote the generating function counting the cardinalities of $L$-trees containing $t$ as a subtree. In the following series of propositions, we derive the closed-form solution for $\overline{\T{L}}(z)$ and check the conditions of~\autoref{th:singularity-analysis} used subsequently to show that in fact $[z^n]\overline{\T{L}}(z) \sim [z^n]\T{L}(z)$, independently of $L$.

\begin{prop}
Let $\T{L}$ be the set of $L$-trees where $|L| = d$. Then its counting sequence ${\set{\T{L,n}}}_{n \in \nats}$ has a corresponding generating function $\T{L}(z)$ given by
\begin{equation}\label{eq:TL(z)-closed-form-solution}
\T{L}(z) = \frac{1-\sqrt{1-4  d z}}{2 z}.
\end{equation}
\end{prop}
 
\begin{proof}
Note that $\T{L}$ can be defined as $\T{L} = L + \T{L} \times \T{L}$,
which translates into the following function equation defining $\T{L}(z)$:
\begin{equation}\label{eq:TL(z)-fun-eq}
\T{L}(z) = d + z {\T{L}(z)}^2.
\end{equation}
Solving~\eqref{eq:TL(z)-fun-eq} for $\T{L}(z)$, we obtain two possible solutions:
\begin{equation*}
\T{L}(z) = \frac{1 \pm \sqrt{1-4  d z}}{2 z}.
\end{equation*}
Since the number of $L$-trees of size $0$ is equal to $d$, the limit $\lim_{z \to 0} \T{L}(z) = d$. It follows that~\eqref{eq:TL(z)-closed-form-solution} is indeed the desired solution.
\end{proof} 
 
\begin{prop}
Let $L$ be a set of $d$ distinct labels. Assume that $t \in \T{L}$ is an $L$-tree of size $p \geq 1$. Then the set of $L$-trees containing $t$ as a subtree, denoted as $\overline{\T{L}}$, has the following generating function:
\begin{equation}\label{eq:overline-TL(z)-closed-form-solution}
\overline{\T{L}}(z) = \frac{-\sqrt{1-4 d z} + \sqrt{1 - 4 d z+4 z^{p+1}}}{2 z}.
\end{equation}
\end{prop}

\begin{proof}
Let us start with noticing that any $L$-tree containing $t$ as a subtree is either equal to $t$, or one of its left or right subtrees contains $t$ whereas the other one is a tree in $\T{L}$. However, since trees in $\T{L}$ may contain $t$ as a subtree, we have to subtract trees containing $t$ in both branches to avoid double counting. This specification allows us to write down the following functional equation defining $\overline{\T{L}}(z)$:
\begin{equation}\label{eq:TL(z)-fixed-term-fun-eq}
\overline{\T{L}}(z) = z^{p} + 2 z  \T{L}(z) \overline{\T{L}}(z) - z {\overline{\T{L}}(z)}^{2}.
\end{equation}
Solving~\eqref{eq:TL(z)-fixed-term-fun-eq} for $\overline{\T{L}}(z)$ we obtain two possible solutions:
\begin{equation*}
\frac{-\sqrt{1-4 d z} \pm \sqrt{1 - 4 d z+4 z^{p+1}}}{2 z}.
\end{equation*}
Note that $p \geq 1$ and hence there are no $L$-trees of size $0$ containing $t$ as a subterm. It follows that $\lim_{z \to 0} \overline{\T{L}}(z) = 0$, yielding the desired solution.
\end{proof}

\begin{prop}\label{prop:T(z)-rho-only-sing}
Let $\zeta = \frac{1}{4d}$. Then $\zeta$ is the only singularity on both $\T{L}(z)$ and $\overline{\T{L}}(z)$'s circle of convergence.
\end{prop}

\begin{proof}
From~\eqref{eq:TL(z)-closed-form-solution} it is clear that $\zeta$ is the only singularity of $\T{L}(z)$ on the circle $|z| < \zeta$. Moreover, since $\sqrt{1 - 4d}$ is a part of $\overline{\T{L}}(z)$'s closed-form expression~\eqref{eq:overline-TL(z)-closed-form-solution}, it suffices to check that $F(z) = 1 - 4 d z+4 z^{p+1}$ has no complex roots of modulus $\zeta$. Note that we can rewrite~\eqref{eq:overline-TL(z)-closed-form-solution} as
\begin{eqnarray*}
\overline{\T{L}}(z) &=& \frac{1 -\sqrt{1-4 d z} - \left( 1 - \sqrt{1 - 4 d z+4 z^{p+1}} \right)}{2 z}\\
&=& \T{L}(z) - \frac{1 - \sqrt{1 - 4 d z+4 z^{p+1}}}{2 z}.
\end{eqnarray*}
Both $\T{L}(z)$ and $\overline{\T{L}}(z)$ are generating functions counting sequences of non-negative integers, hence the coefficients in the Maclaurin series of $\frac{1 - \sqrt{1 - 4 d z+4 z^{p+1}}}{2 z}$ are non-negative integers as well. To finish the proof we notice that 
\begin{equation*}
F(\zeta) = 4^{-p} \left(\frac{1}{d}\right)^{p+1} > 0
\end{equation*}
and hence due to~\hyperref[th:pringsheim]{Pringsheim's Theorem}, $F(z)$ cannot have complex roots of modulus $\zeta$.
\end{proof}

The generating functions $\T{L}(z)$ and $\overline{\T{L}}(z)$ are not defined at $z = 0$, however due to~\autoref{th:riemann-rem-sing}, both have analytic extensions to functions analytic in the origin and we can consider them instead of $\T{L}(z)$ and $\overline{\T{L}}(z)$ in the subsequent theorem.

\begin{theorem}\label{th:TL(z)-asymptotic-approx}
Both $[z^n]\T{L}(z)$ and $[z^n]\overline{\T{L}}(z)$ admit for large $n$ the following asymptotic approximation:
\begin{equation}\label{eq:TL(z)-asymptotic-approx}
[z^n]\T{L}(z) \sim [z^n]\overline{\T{L}}(z) \sim {(4 d)}^n \frac{-2d n^{-\nicefrac{3}{2}}}{\Gamma(-\frac{1}{2})}.
\end{equation}
\end{theorem}

\begin{proof}
Let us rewrite the closed-form solutions of $\T{L}(z)$~\eqref{eq:TL(z)-closed-form-solution} and $\overline{\T{L}}(z)$~\eqref{eq:overline-TL(z)-closed-form-solution} as
\begin{eqnarray*}
\T{L}(z) &=& \sqrt{1-4 d}\left( -\frac{1}{2z} \right) + \frac{1}{2z}, \quad \text{and}\\
\overline{\T{L}}(z) &=& \sqrt{1-4 d}\left( -\frac{1}{2z} \right) + \frac{\sqrt{1 - 4 d z+4 z^{p+1}}}{2 z}.
\end{eqnarray*}
As all the assumptions hold, the result follows now easily by applying~\autoref{th:singularity-analysis}.
\end{proof}

Immediately, we obtain the following corollary theorem.
\begin{theorem}
Let $t \in T_L$. Then asymptotically almost all $L$-trees contain $t$ as a subtree.
\end{theorem}

\begin{proof}
In the case of $|t| \geq 1$, our claim follows directly from~\autoref{th:TL(z)-asymptotic-approx}. Suppose that $|t| = 0$, i.e.~$t$ is a primitive combinator. We can assume that $|L| > 1$, as otherwise our claim is trivial. Let us consider the set $\T{L \setminus \set{t}}$ of $L$-trees avoiding $t$. Note that from~\eqref{eq:TL(z)-closed-form-solution}, we have $[z^n]\T{L \setminus \set{t}}(z) \sim c_1 4^n(d-1)^n n^{-3/2}$, whereas $[z^n]\T{L}(z) \sim c_2 (4d)^n n^{-3/2}$ for some constants $c_1$ and $c_2$. It follows that asymptotically almost no $L$-tree avoids the primitive combinator $t$, finishing the proof.
\end{proof}

The above theorem provides a general result showing that `local' properties of $L$-trees propagating to supertrees, span asymptotically almost the whole set of $L$-trees. Using the natural bijection between $\B$-combinators and $\B$-trees, we can reinterpret this observation in the language of combinatory logic and state that each `local' property of $\B$-combinators closed under taking superterms is typical, i.e.~has asymptotic probability $1$ in the set of all $\B$-combinators. In particular, we obtain the following corollaries generalising the results in~\cite{dgkrtz}.

\begin{cor}\label{cor:subterm-prop}
For each Turing-complete basis of primitive combinators $\B$, asymptotically almost no $\B$-combinator is in normal form, simply typeable nor strongly normalising.
\end{cor}

\begin{proof}
Fix $t := \overline{\comb{S}} \comb{I} \comb{I} (\overline{\comb{S}} \comb{I} \comb{I})$ where $\comb{I} := \overline{\comb{S} \comb{K} \comb{K}}$.
\end{proof}

\subsection{Counting normalising combinators}\label{subsec:counting-normalizing-combinators}
Let $\B$ be a Turing-complete set of primitive combinators. Let us start with the classical observation is that the set of normalising $\B$-combinators is undecidable. It follows that its corresponding generating function has no computable closed-form solution. For that reason we take the following approach. We find feasible subclasses of normalising and non-normalising $\B$-combinators and use them to bound the density of normalising combinators $\WN_\B$. Let us start with recalling the famous standardisation theorem.

\begin{theorem}[Standardisation theorem, see e.g.~\cite{curry-feys}]\label{th:standardization-theorem}
If $M$ is normalising, then the leftmost outermost reduction always leads to $M$'s normal form.
\end{theorem}

Although usually stated in the $\S \K$-basis, the standardisation theorem easily generalises to every Turing-complete basis of combinators, e.g.~through the classical translation to $\lambda$-calculus (see, e.g.~\cite{BAR84}). Hence, we obtain the following result.

\begin{theorem}
Let $\B$ be a Turing-complete basis of primitive combinators. Then
\begin{equation*}
0 < \densityminus{\WN_\B}{\class{B}} \quad \text{and} \quad \densityplus{\WN_\B}{\class{B}} < 1.
\end{equation*}
\end{theorem}

\begin{proof}
Let us start with the lower bound. Since $\B$ is Turing-complete, there exists a combinator $\overline{\comb{K}} \in \class{\B}$ extensionally equivalent to $\comb{K}$. Let us consider the set $L$ of combinators in form of $\overline{\comb{K}} \comb{X} M$ where $\comb{X} \in \B$ is a primitive combinator and $M \in \class{\B}$. Notice that if $t \in L$, then $t$ has a normal form. Hence $L \subset \WN_\B$. Let us fix $p := |\overline{\comb{K}}|$. Then
\begin{eqnarray*}
\density{L}{\class{\B}} &=& \lim_{n \to \infty} \frac{|L_n|}{|\class{\B,n}|} = \lim_{n \to \infty} \frac{d \cdot |\class{\B,n-p-2}|}{|\class{\B,n}|}\\
&=& \lim_{n \to \infty} \frac{d^{n-p} \cdot C_{n-p-2}}{d^{n+1} \cdot C_n} =
\frac{1}{d^{p+1} \cdot 4^{p+2}} > 0.
\end{eqnarray*}
Naturally we have
\begin{equation*}
\density{L}{\class{\B}} \leq \densityminus{\WN_\B}{\class{\B}}
\end{equation*}
hence indeed, the lower bound holds.

Now, let us consider the upper bound. Let $\omega = \overline{\comb{S} \comb{I} \comb{I}}$. Since $\comb{S} \comb{I} \comb{I} x \to_w x x$, we know that $\comb{S} \comb{I} \comb{I} (\comb{S} \comb{I} \comb{I})$ has no normal form. Immediately, nor does $\omega \omega$. Consider the transformation $\Phi \colon \class{\B} \to \class{\B}$ which for a given $\B$-combinator substitutes $\omega \omega$ for its leftmost primitive combinator $\comb{X}$ (see~\autoref{fig:U-transform}).

\begin{figure}[th!]
\resizebox{.75\linewidth}{!}{
\begin{subfigure}{.3\textwidth}
	  \centering
\begin{tikzpicture}[triangle/.style = {regular polygon, regular polygon sides=3 },
level1/.style ={level distance=1cm},]
\node [circle,draw,fill=black,scale=0.5] (a) {}
    child[level1] {node [circle,draw,fill=black,scale=0.5] (b) {}
      child[level1] {node {$\vdots$}
        child[level1] {node [label={[yshift=-8mm]$\comb{X}$},circle,draw,fill=black,scale=0.5] (d) {}}
        child[level1] {
      	node {$\vdots$}
      	}
      } 
      child[level1] {
      node {$\vdots$}
      }
    }
    child[level1] {
      node {$\vdots$}
      };
\end{tikzpicture}
\end{subfigure}
\begin{subfigure}{.2\textwidth}
\[ \mathrel{\overset{\makebox[0pt]{\mbox{$\Phi$}}}{\longmapsto}} \]
\end{subfigure}
\begin{subfigure}{.3\textwidth}
	  \centering
	   \tikzset{
itria/.style={
  draw,shape border uses incircle, scale=0.75,
  isosceles triangle,shape border rotate=90,yshift=-1.45cm}
  }	  
	  
	  \begin{tikzpicture}[triangle/.style = {regular polygon, regular polygon sides=3 },
level1/.style ={level distance=1cm},]
\node [circle,draw,fill=black,scale=0.5] (a) {}
    child[level1] {node [circle,draw,fill=black,scale=0.5] (b) {}
      child[level1] {node {$\vdots$}
        child[level1] { 
		node [circle,draw,fill=black,scale=0.5] (q) {}        
        node[itria] {$\comb{\omega \omega}$}} 
        child[level1] {
      	node {$\vdots$}
      	}
      } 
      child[level1] {
      node {$\vdots$}
      }
    }
    child[level1] {
      node {$\vdots$}
      };
\end{tikzpicture}
\end{subfigure}
}
\caption{Transformation $\Phi$}
\label{fig:U-transform}
\end{figure}
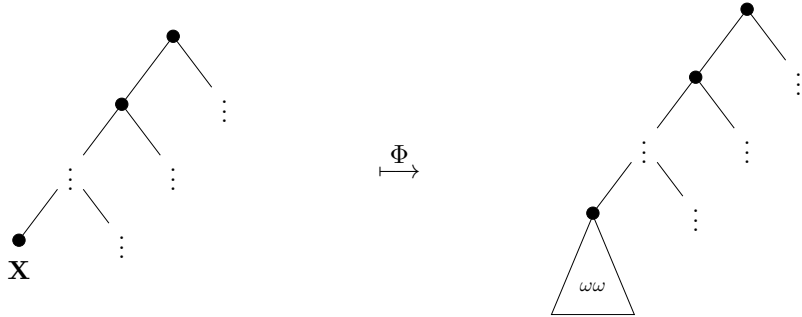

Let $U$ be the image of $\class{\B}$ through $\Phi$. By~\autoref{th:standardization-theorem}, we know that $U$ is a set of non-normalising combinators. In other words, we have $\WN_\B \subset \class{\B} \setminus U$. Let $M$ be an arbitrary combinator in $U$. Note that since there are $d$ primitive combinators, the map $\Phi$ sends exactly $d$ distinct combinators to $M$. For convenience, let us set $p := |\omega \omega|$. Then we obtain
\begin{eqnarray*}
\density{U}{\class{\B}} &=& \lim_{n \to \infty} \frac{|U_n|}{|\class{\B,n}|} = \lim_{n \to \infty} \frac{d \cdot |\class{\B,n-p}|}{|\class{\B,n}|}\\
&=& \lim_{n \to \infty} \frac{d^{n-p+2} \cdot C_{n-p}}{d^{n+1} \cdot C_n} =
\frac{d}{{(4 d)}^{p}} > 0.
\end{eqnarray*}
Since $\WN_\B \subset \class{\B} \setminus U$, we have
\begin{equation*}
\densityplus{\WN_\B}{\class{\B}} \leq 1 - \density{U}{\class{\B}},
\end{equation*}
and thus the upper bound holds as well, finishing the proof.
\end{proof}

Using the fact that asymptotically no $\B$-combinator is strongly normalising (see~\autoref{cor:subterm-prop}), we obtain the following corollary generalising the result in~\cite{Bendkowski2015}.

\begin{cor}
For each Turing-complete basis $\B$ of primitive combinators, asymptotically every weakly normalising $\B$-combinator is not strongly normalising.
\end{cor}

\section{SK-combinators}\label{sec:sk-combinators}
In this section we address the problem of estimating the asymptotic density of normalising $\comb{S}\comb{K}$-combinators in the set of all $\comb{S}\comb{K}$-combinators. In~\cite{Bendkowski2015} authors provided the following bounds:
\begin{equation*}
\frac{1}{32} \leq \densityminus{\WN_{\comb{S}\comb{K}}}{\class{\comb{S}\comb{K}}} \quad \text{and} \quad \densityplus{\WN_{\comb{S}\comb{K}}}{\class{\comb{S}\comb{K}}} \leq 1 - \frac{1}{2^{18}}.
\end{equation*}

Here, we prove that $\comb{S}\comb{K}$-combinators reducing in exactly $n>0$ normal-order reduction steps (leftmost outermost redex first, see~\autoref{th:standardization-theorem}) have positive asymptotic density in the set of all $\comb{S}\comb{K}$-combinators. We provide a constructive method of finding their asymptotic approximations and related densities, yielding a systematic approach to improving the above lower bound. For simplicity, we use $\class{\comb{S}\comb{K}}$ and $\rg{0}$ to denote the set of $\comb{S}\comb{K}$-combinators and the set of normal forms, respectively. Let us start with a few technical propositions regarding the generating functions $C(z)$ and $\rgf{0}$.

\begin{prop}[see, e.g.~\cite{Bendkowski2015}]
The generating function $C(z)$ enumerating $\S \K$-terms and its corresponding dominating singularity $\domsing{C}$ are given by
\begin{equation}\label{eq:c(z)}
C(z) = \frac{1-\sqrt{1-8z}}{2z} \quad \text{and} \quad \domsing{C} = \frac{1}{8}.
\end{equation}
\end{prop}

\begin{prop}\label{prop:c(z)-properties}
Let $n \geq 1$. Then ${C(z)}^n = \sqrt{1-8z}\, P(z) + Q(z)$ for some rational functions $P(z)$ and $Q(z)$ analytic in $\mathbb{C} \setminus \set{0}$. Moreover, ${C(z)}^n$ has a single removable singularity at $z = 0$ in the disk $|z| < \domsing{C}$.
\end{prop}

\begin{proof}
From equation~\eqref{eq:c(z)}, $C(z)$ can be rewritten as \[ C(z) = \sqrt{1-8z}\, P(z) + Q(z) \quad \text{where} \quad P(z) = -\frac{1}{2z} \quad \text{and} \quad Q(z) = \frac{1}{2z}. \]
Both $P(z)$ and $Q(z)$ are rational and hence also analytic in the complex plane except the origin. From~\autoref{lem:sqrt-commutative-ring}, ${C(z)}^n$ can be expressed as
\[ {C(z)}^n = \sqrt{1-8z}\, \overline{P}(z) + \overline{Q}(z) \]
for some $\overline{P}(z)$ and $\overline{Q}(z)$ analytic in $\mathbb{C} \setminus \set{0}$. Moreover, following~\eqref{eq:sqrt-mult-eq} and the closure properties of rational functions, it is clear that $\overline{P}(z)$ and $\overline{Q}(z)$ are also rational. As $\lim_{z \to 0} z C(z) = 0$, \autoref{th:riemann-rem-sing} guarantees that $C(z)$ has an analytic extension at $z = 0$ and, in consequence of~\autoref{lem:sing-powers}, so does ${C(z)}^n$, finishing the proof.
\end{proof}

\begin{prop}[see~\cite{Bendkowski2015}]
The generating function $\rgf{0}$ enumerating $\S \K$-terms in normal form and its corresponding dominating singularity $\domsing{0}$ are given by
\begin{equation}\label{eq:r0(z)}
 \rgf{0} = \frac{1 - 2z - \sqrt{1-4z-4z^2}}{2z^2} \quad \text{and} \quad \domsing{0} = \frac{1}{2}\left(\sqrt{2}-1\right) \approx 0.207107.
 \end{equation}
\end{prop}

\begin{prop}\label{prop:r0(z)-properties}
Let $n \geq 1$. Then ${R_0(z)}^n = \sqrt{1-4z-4z^2}\, P(z) + Q(z)$ for some rational functions $P(z)$ and $Q(z)$ analytic in $\mathbb{C} \setminus \set{0}$. Moreover, ${R_0(z)}^n$ has a single removable singularity at $z = 0$ in the disk $|z| < \domsing{0}$.
\end{prop}

\begin{proof}
From the shape of equation~\eqref{eq:r0(z)} we can write \[ R_0(z) = \sqrt{1-4z-4z^2}\, P(z) + Q(z) \quad \text{where} \quad P(z) = -\frac{1}{2z^2} \quad \text{and} \quad Q(z) = \frac{1-2z}{2z^2}. \]
Clearly, both $P(z)$ and $Q(z)$ are rational and analytic in $\mathbb{C} \setminus \set{0}$. The result follows now from the same arguments as in~\autoref{prop:c(z)-properties}.
\end{proof}

Our method relies on the effective computation and asymptotic properties of \emph{normal-order reduction grammars} $\set{\rg{n}}_{n\in\nats}$. In~\cite{Bendkowski16} the author provided a recursive algorithm, which for given $n \geq 1$, constructs the $n$th normal-order reduction grammar $\rg{n}$
defining the set of $\S \K$-combinators reducing in exactly $n$ normal-order reduction steps. Applying the Symbolic Method of Flajolet and Sedgewick~\cite{FlajoletSedgewick2009}, $\rg{n}$ is then translated into a functional equation
\begin{equation}\label{eq:rgf-gen-fun-eq-psi}
\rgf{n} = \Psi(\rgf{n},z)
\end{equation}
involving the generating function $\rgf{n}$ and its formal parameter $z$. Due to the specific structure of $\Phi(\rg{n})$ -- the set of productions $\alpha \in \rg{n}$ not referencing $\rg{n}$ --~\eqref{eq:rgf-gen-fun-eq-psi} turns out to be linear in $\rgf{n}$, yielding a unique closed-form solution.

\begin{theorem}[see~\cite{Bendkowski16}]\label{th:reduction-grammars}
Let $n \geq 0$. Then there exists a computable unambiguous regular tree grammar $\rg{n}$ defining the set of $\S \K$-combinators reducing in exactly $n$ normal-order reduction steps. Moreover, $\rg{n}$ has a computable generating function $\rgf{n}$ of the following closed-form solution:
\begin{equation}\label{eq:rn(z)-fun-eq-simplified}
\rgf{n} = \frac{1}{\sqrt{1-4 z - 4 z^2}} \sum_{\alpha \in \Phi(\rg{n})} R_{\alpha}(z),
\end{equation}
where
\begin{equation}\label{eq:ralpha-fun-eq}
R_{\alpha}(z) = z^{k(\alpha)} {C(z)}^{c(\alpha)} \prod_{i=0}^{n-1} {R_i(z)}^{r_i(\alpha)}
\end{equation}
and $k(\alpha)$, $c(\alpha)$, $r_i(\alpha)$ are some non-negative integers depending on $\alpha$.
\end{theorem}

In the reminder of this section, we exploit the structure of the normal-order reduction grammars, showing the following main result.

\begin{theorem}\label{th:main}
Let $k \geq 1$. Then the asymptotic growth rate of $[z^n]\rgf{k}$ is given by
\[ [z^n]\rgf{k} \sim 8^n \frac{\overline{C}_k n^{-\nicefrac{3}{2}}}{\Gamma(-\nicefrac{1}{2})}, \]
where $\overline{C}_k$ is a constant depending on $k$.
\end{theorem}

Before we provide a proof, let us present two propositions preparing the background for~\autoref{th:singularity-analysis}.

\begin{prop}\label{prop:sing-at-0}
Let $n \geq 0$. Then each $\rgf{n}$ has a removable singularity at $z = 0$.
\end{prop}

\begin{proof}
Induction over $n$. Following~\autoref{th:riemann-rem-sing}, $\rgf{n}$ has a removable singularity at $z = 0$ if and only if the limit $\lim_{z \to 0} z \rgf{n}$ exists and is equal to $0$. In particular, from~\eqref{eq:rn(z)-fun-eq-simplified} 
\begin{equation}\label{eq:rn(z)-rem-sing-form}
\lim_{z \to 0} \frac{z}{\sqrt{1-4 z - 4 z^2}} \bigg( z^{k(\alpha)} {C(z)}^{c(\alpha)} \prod_{i=0}^{n-1} {R_i(z)}^{r_i(\alpha)} \bigg) = 0
\end{equation}
for each $\alpha \in \Phi(\rg{n})$. 

Let us start with $n = 0$. In this case the product $\prod_{i=0}^{n-1} {R_i(z)}^{r_i(\alpha)}$ vanishes, simplifying~\eqref{eq:rn(z)-rem-sing-form} to
\begin{equation}\label{eq:rn(z)-rem-sing-form-2}
\lim_{z \to 0} \frac{z^{k(\alpha)+1} {C(z)}^{c(\alpha)}}{\sqrt{1-4 z - 4 z^2}} = 0.
\end{equation}
Due to~\autoref{prop:r0(z)-properties}, ${C(z)}^{c(\alpha)}$ has an analytic extension at $z = 0$. It follows that $\lim_{z \to 0} {C(z)}^{c(\alpha)}$ exists, indeed satisfying equation~\eqref{eq:rn(z)-rem-sing-form-2}.

Now, suppose that $n > 0$. Due to the induction hypothesis all $\rgf{0},\ldots,\rgf{n-1}$ have removable singularities at $z = 0$. Using~\autoref{lem:sing-powers}, we can moreover state that so do their powers ${\rgf{0}}^{r_0(\alpha)},\ldots,{\rgf{n-1}}^{r_{n-1}(\alpha)}$. Together with our previous observation that ${C(z)}^{c(\alpha)}$ has an analytic extension at $z = 0$, we conclude that~\eqref{eq:rn(z)-rem-sing-form} is satisfied, finishing the proof.
\end{proof}

\begin{definition}
Let $\alpha \in \Phi(\rg{n})$ for some $n \geq 1$. We say that $\alpha$ is \emph{major} if and only if $\alpha$ references either $\class{\comb{S} \comb{K}}$ or some $\rg{i}$ for $i \in \set{1,\ldots,n-1}$. Otherwise, we say that $\alpha$ is \emph{minor}. 
\end{definition}

In the following proposition we use the notions of major and minor productions, showing that major productions contribute to the asymptotic growth rate of $\rgf{n}$'s counting sequence, whereas minor ones are asymptotically negligible.

\begin{prop}\label{prop:form-of-rn(z)}
Let $n \geq 1$. Then each $\rgf{n}$ is in form of $\sqrt{1-8z}\, P(z) + Q(z)$ where both $P(z)$ and $Q(z)$ are analytic in the disk $|z| < \zeta_0 = \frac{\sqrt{2}-1}{2}$ but at $z = 0$.
\end{prop}

\begin{proof}
Induction over $n$. Consider the base case $n=1$. 
Let us divide $\Phi(\rg{1})$ into two groups, i.e.~major and minor productions. Suppose that $\alpha \in \rg{1}$ is a major production. Since $\alpha \in \Phi(\rg{1})$, its corresponding generating function~\eqref{eq:ralpha-fun-eq} $\rgf{\alpha}$ is in form of
\begin{equation}\label{eq:rn(z)-asymptotic-form-1}
\rgf{\alpha} = z^{k(\alpha)} {C(z)}^{c(\alpha)} {R_0(z)}^{r_0(\alpha)},
\end{equation}
where in addition $c(\alpha) \geq 1$. Using Propositions~\ref{prop:c(z)-properties} and~\ref{prop:r0(z)-properties}, we can further rewrite~\eqref{eq:rn(z)-asymptotic-form-1} as
\begin{equation*}
\rgf{\alpha} = \sqrt{1-8z}\, \overline{P}(z) + \overline{Q}(z)
\end{equation*}
for functions $\overline{P}(z)$ and $\overline{Q}(z)$ analytic in the disk $|z| < \zeta_0$ but at $z=0$. Similarly, if $\alpha \in \Phi(\rg{1})$ is minor, we can rewrite its generating function as
\begin{equation*}
\rgf{\alpha} = \sqrt{1-4z-4z^2}\, \widehat{P}(z) + \widehat{Q}(z)
\end{equation*}
where $\widehat{P}(z)$ and $\widehat{Q}(z)$ are analytic in some disk $|z| < \zeta_0 + \varepsilon$ for $\varepsilon > 0$ but at $z=0$. The requested form of $\rgf{1}$ follows now from~\autoref{lem:sqrt-commutative-ring} and the fact that $\zeta_1 < \zeta_0$.

Now, suppose that $n > 1$. Again, let us consider an arbitrary major $\alpha \in \Phi(\rg{n})$. Using the induction hypothesis and~\eqref{eq:ralpha-fun-eq}, we can rewrite $\rgf{\alpha}$ as
\begin{equation}\label{eq:rn(z)-asymptotic-form-2}
\rgf{\alpha} = z^{k(\alpha)} {C(z)}^{c(\alpha)} {\rgf{0}}^{r_0(\alpha)} \prod_{i=1}^{n-1} {\bigg(\sqrt{1-8z}\, \overline{P_i}(z) + \overline{Q_i}(z)\bigg)}^{r_i(\alpha)}.
\end{equation}

Using Propositions~\ref{prop:c(z)-properties} and~\ref{prop:r0(z)-properties}, we can further rewrite~\eqref{eq:rn(z)-asymptotic-form-2} as
\begin{equation}\label{eq:rn(z)-asymptotic-form-3}
\rgf{\alpha} = \bigg(\sqrt{1-8z}\, \overline{P}(z) + \overline{Q}(z)\bigg) \prod_{i=1}^{n-1} {\bigg(\sqrt{1-8z}\, \overline{\overline{P_i}}(z) + \overline{\overline{Q_i}}(z)\bigg)}^{r_i(\alpha)}.
\end{equation}
The result follows now easily from~\autoref{lem:sqrt-commutative-ring}.
\end{proof}

Now we are in a position to prove \autoref{th:main}.

\begin{proof}{(\autoref{th:main})}
Let $k>0$. Due to \autoref{prop:form-of-rn(z)}, every function $R_k(z)$ is in form of $\sqrt{1-8z}P_k(z) + Q_k(z)$ for some algebraic functions $P_k(z)$ and $Q_k(z)$ that are analytic in the disk $|z|<\frac{\sqrt{2}-1}{2}$ but at $z=0$. By \autoref{prop:sing-at-0}, every $R_k(z)$ has a removable singularity at $z=0$. Therefore, every function $R_k(z)$ satisfies the assumptions of \autoref{th:singularity-analysis}. Hence
\[ [z^n]\rgf{k} \sim 8^n \frac{\overline{C}_k n^{-\nicefrac{3}{2}}}{\Gamma(-\nicefrac{1}{2})}\]
for some constant $\overline{C}_k$.
\end{proof}

As $\domsing{m} = \nicefrac{1}{8}$ for every $m \geq 1$, we can easily compute the coefficients $\overline{C}_m$ in the asymptotic approximation of $[z^n]\rgf{m}$ using available computer algebra systems, e.g. Mathematica \textregistered~\cite{mathematicaSoft}. The quotient $\nicefrac{\overline{C}_m}{-4}$ (see~\autoref{eq:TL(z)-asymptotic-approx}) yields the desired asymptotic density of $\S \K$-combinators normalising in $m$ normal-order reduction steps in the set of all $\S \K$-combinators. Hence, we obtain the following corollary.

\begin{cor}
For each $m \in \nats$, finding the asymptotic density of combinators normalising in $m$ steps in the set of all $\S \K$-combinators is computable.
\end{cor}

Using an implementation of the normal-order reduction grammar algorithm together with Mathematica \textregistered~we were able to compute the densities of combinators reducing in $m$ normal-order reduction steps $\rg{m}$ in $C_{\S \K}$ for $m = 1,\ldots,7$. The results are summarised in the following figure.
\begin{figure}[th!]
\begin{displaymath}
        \begin{array}{r|l}
        m & \density{\rg{m}}{C_{\S \K}} \\\hline
        1 & 0.08961233291075565\\\hline
        2 & 0.06417374404832035\\\hline
        3 & 0.0501056553007704\\\hline
        4 & 0.04131967414765603\\\hline
        5 & 0.03570996929825453\\\hline
        6 & 0.03119525702124082\\\hline
        7 & 0.027987393260263862\\
        \end{array}
\end{displaymath}
\caption{Density of $\rg{m}$ in $\class{\S \K}$ for $m=1,\ldots,7$}
\end{figure}

And so, exploiting the finite additivity of asymptotic density we obtain the following improved lower bound:
\begin{equation*}
0.34010402598726164 \leq \densityminus{\WN_{\comb{S}\comb{K}}}{\class{\comb{S}\comb{K}}}.
\end{equation*}

Clearly, the above lower bound can be further improved if we compute the next asymptotic densities $\density{\rg{8}}{\class{\S \K}},\density{\rg{9}}{\class{\S \K}},\density{\rg{10}}{\class{\S \K}}$, etc. Unfortunately, due to the sheer amount of major productions, this process is immensely time and memory consuming, quickly requiring resources exceeding current desktop computer capabilities. Nevertheless, $\density{\rg{m}}{\class{\S \K}} \leq 1$ for each $m$ and hence
$\sum_{m \geq 0} \density{\rg{m}}{\class{\S \K}}$ is necessarily convergent to some value $0 < \zeta < 1$. Moreover, if $\density{\WN_{\comb{S}\comb{K}}}{\class{\S \K}}$ exists, then we have $\zeta \leq \density{\WN_{\comb{S}\comb{K}}}{\class{\S \K}}$. The intriguing problem of determining whether the inequality can be replaced by the equality still remains open.

\section{Experimental results}\label{sec:experimental-results}
Our method developed in~\autoref{sec:sk-combinators} allows us to improve the lower bound on $\densityminus{\WN_{\comb{S}\comb{K}}}{\class{\comb{S}\comb{K}}}$, provided we have enough computational resources to find and manipulate generating functions $\rgf{m}$ for higher $m$. Unfortunately, the current gap between the lower and upper bound on the density of normalising combinators is still quite significant. In this section we present some experimental results regarding the aforementioned density as well as numerical evaluations of the obtained approximation error.

\subsection{Super-computer results}\label{subsec:super-computer-results}
Consider the following experiment scheme $G(s,n,r)$ with three positive integer parameters $s,n,r$. We draw $s$ uniformly random $\S \K$-combinators of size $n$ using, e.g.~R{\'e}my's algorithm (see~\cite{DBLP:journals/ita/Remy85,Knuth:2006:ACP:1121689}) -- drawing a uniformly random plane binary tree with $n$ inner nodes -- combined with a random $\S \K$-labelling. Then we reduce each of the $s$ samples using up to $r$ normal-order reduction steps. We record then the number of normalised samples, with their corresponding reduction lengths. For samples which did not normalise in $r$ reduction steps, we artificially record their reduction lengths as $-1$. Afterwards, we collect the reduction lengths, plotting the obtained function mapping reduction lengths to the number of samples attaining the given reduction length.

We preformed our experiments on the \emph{Prometheus \textregistered}~super-computer cluster granted by ACC Cyfronet AGH in Kraków, Poland (\nth{48} out of 500 world's most powerful supercomputer in 2016, with theoretical computational power of 2.4 Pflops)~\cite{prometheus}. The following figure summarises the experiment result for $G(s = 1200,n = 50000000,r = 1000)$.
\begin{figure}[th!]\label{fig:experiment}
      \centering
      \resizebox {200px} {!} {
    \begin{tikzpicture}
\begin{axis}[
    xlabel={normal-order reductions},
    ylabel={Number of samples},
    xmin=-1, xmax=868,
    ymin=0, ymax=176,
    legend pos=north west,
    ymajorgrids=true,
    grid style=dashed,
]
 
\addplot[
    color=blue,
    mark=*,
    ]
    coordinates {
(-1,176)
(1,110)
(2,92)
(3,57)
(4,47)
(5,33)
(6,24)
(7,35)
(8,30)
(9,21)
(10,25)
(11,25)
(12,33)
(13,19)
(14,17)
(15,16)
(16,13)
(17,12)
(18,10)
(19,18)
(20,18)
(21,12)
(22,12)
(23,6)
(24,12)
(25,15)
(26,8)
(27,9)
(28,10)
(29,8)
(30,8)
(31,12)
(32,10)
(33,6)
(34,4)
(35,8)
(36,9)
(37,6)
(38,9)
(39,10)
(40,2)
(41,5)
(42,9)
(43,3)
(44,5)
(45,3)
(46,8)
(47,8)
(48,4)
(49,4)
(50,2)
(51,1)
(52,1)
(53,4)
(54,4)
(55,3)
(56,3)
(57,4)
(58,4)
(59,2)
(60,1)
(61,8)
(62,2)
(63,2)
(64,1)
(65,2)
(66,2)
(67,1)
(68,2)
(69,1)
(70,2)
(72,2)
(75,1)
(76,2)
(77,1)
(78,1)
(79,2)
(81,2)
(83,2)
(84,2)
(85,1)
(89,2)
(90,1)
(91,1)
(93,2)
(95,3)
(96,3)
(97,2)
(98,1)
(100,2)
(102,1)
(103,1)
(108,1)
(109,2)
(110,1)
(111,3)
(112,1)
(114,3)
(116,1)
(119,1)
(120,1)
(126,3)
(129,1)
(135,1)
(136,1)
(137,1)
(141,1)
(167,1)
(170,1)
(183,2)
(192,1)
(193,1)
(195,1)
(208,2)
(216,1)
(218,1)
(219,1)
(223,1)
(234,1)
(236,1)
(237,1)
(245,1)
(279,2)
(282,1)
(289,1)
(290,1)
(301,1)
(314,1)
(316,1)
(329,1)
(401,1)
(469,1)
(478,1)
(490,1)
(601,1)
(684,1)
(693,1)
(698,1)
(868,1)
    };
    \legend{samples}
 
\end{axis}
\end{tikzpicture}}
	\caption{$G(1200, 50000000, 1000)$}
    \end{figure}
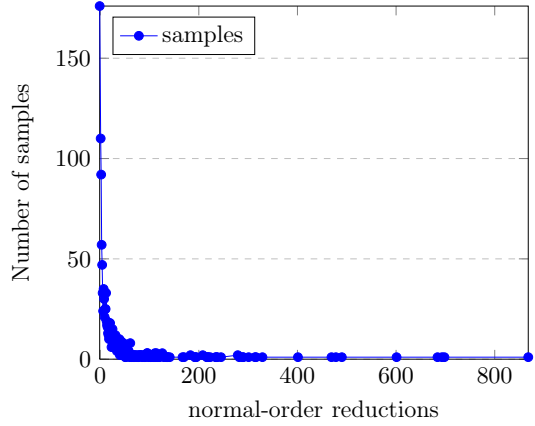

Even though the number $r = 1000$ bounding the number of performed reductions was significantly smaller than the size of considered samples, the experiment revealed that normalising combinators have short reduction lengths. In fact, the mean reduction length $E(X)$ over all normalising samples is approximately
$E(X) \approx 31.5810$, whereas $\log_2 n \approx 25.5754$. Out of $1200$ samples, only $176$ did not normalise in $r$ steps, yielding approximately $0.146$ percent of all considered samples. Similar results were obtained with different experiment parameters, suggesting that the ratio of normalising $\S \K$-combinators should be approximately equal to $85\%$, whereas the mean reduction length of normalising terms is $\Theta(\log_2 n)$.
Our Haskell implementation of the program, as well as all the obtained data sets are available at~\cite{mb-haskell-implementation-experiments}.

\subsection{Approximation error}\label{subsec:approximation-error}
In~\autoref{sec:sk-combinators} we showed that $[z^n]\rgf{m} = \Theta([z^n]C(z))$ for $m \geq 1$. Nonetheless, techniques used to obtain this result do not yield the convergence rate of sequences in question. Using Mathematica \textregistered~we compared $[z^n]\rgf{m}$ with its approximation $8^n \widetilde{C}_m n^{-\nicefrac{3}{2}}$. The following figures summarise values for $[z^n]\rgf{1}$ and their relative error $\delta([z^n]\rgf{1})$.

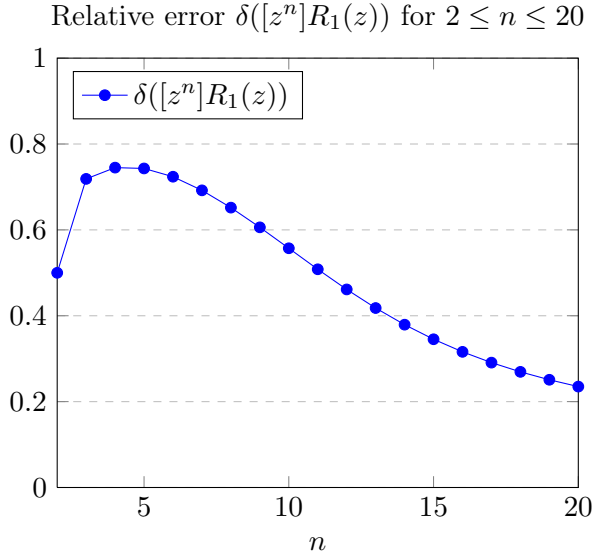
\begin{figure}[th!]
      \centering
      \begin{displaymath}
        \begin{array}{r|r|r|r|r}
          n & [z^n]C(z) & [z^n]\rgf{1} & \lfloor 8^n \widetilde{C}_1 n^{-\nicefrac{3}{2}} \rfloor & \delta([z^n]\rgf{1}) \\\hline
          2 & 16 & 4 & 2 & 0.5 \\
3 & 80 & 32 & 9 & 0.71875 \\
4 & 448 & 200 & 51 & 0.745 \\
5 & 2688 & 1152 & 296 & 0.7430555555555556 \\
6 & 16896 & 6528 & 1803 & 0.7238051470588235 \\
7 & 109824 & 37184 & 11450 & 0.6920718588640276 \\
8 & 732160 & 215328 & 74973 & 0.6518195497102095 \\
9 & 4978688 & 1275520 & 502653 & 0.6059230745107878 \\
10 & 34398208 & 7753472 & 3433386 & 0.5571808345990029 \\
11 & 240787456 & 48412416 & 23808041 & 0.5082244810917926 \\
12 & 1704034304 & 310294272 & 167159405 & 0.4612874935699748 \\
13 & 12171673600 & 2037696512 & 1185980764 & 0.4179796858777761 \\
14 & 87636049920 & 13675532288 & 8489666053 & 0.3792076334425748 \\
15 & 635361361920 & 93532264448 & 61240081391 & 0.3452518042579103 \\
16 & 4634400522240 & 650108973568 & 444715903783 & 0.3159363708790836 \\
17 & 33985603829760 & 4580578080768 & 3248472837654 & 0.2908159668114757 \\
18 & 250420238745600 & 32644683026432 & 23852497067944 & 0.2693298002424797 \\
19 & 1853109766717440 & 234890688573440 & 175955235773882 & 0.2509058709712602 \\
20 & 13765958267043840 & 1703833526784000 & 1303399617705108 & 0.2350193858637788
        \end{array}
      \end{displaymath}
      \caption{$[z^n]\rgf{1} \sim 8^n \widetilde{C}_1 n^{- \nicefrac{3}{2}}$ with $\widetilde{C}_1 \approx 0.10111668957132425$.}
    \end{figure}

\begin{figure}[th!]
      \centering
      \begin{subfigure}{.6\textwidth}
	  \centering
    \begin{tikzpicture}
\begin{axis}[
    title={Relative error $\delta([z^n]\rgf{1})$ for $2 \leq n \leq 20$},
    xlabel={$n$},
    xmin=2, xmax=20,
    ymin=0, ymax=1,    
    legend pos=north west,
    ymajorgrids=true,
    grid style=dashed,
]
 
\addplot[
    color=blue,
    mark=*,
    ]
    coordinates {
    (2,0.5)
	(3, 0.71875)
	(4, 0.745)
    (5,0.7430555555555556)
	(6,0.7238051470588235)
	(7,0.6920718588640276)
	(8,0.6518195497102095)
	(9,0.6059230745107878)
	(10,0.5571808345990029)
	(11,0.5082244810917926)
	(12,0.4612874935699748)
	(13,0.4179796858777761)
	(14,0.3792076334425748)
	(15,0.3452518042579103)
	(16,0.3159363708790836)
	(17,0.29081596681147576)
	(18,0.26932980024247977)
	(19,0.2509058709712602)
	(20,0.2350193858637788)
    };
    \legend{$\delta([z^n]\rgf{1})$}
 
\end{axis}
\end{tikzpicture}
\end{subfigure}%
	  \begin{subfigure}{.4\textwidth}
	  \centering
  		\begin{displaymath}
        \begin{array}{r|l}
          n & \delta([z^n]\rgf{1}) \\\hline
          20 & 0.2350193858637788 \\
40 & 0.10914523244529438 \\
60 & 0.07194386101679777 \\
80 & 0.053697148923235606 \\
100 & 0.042841972392710106 \\
120 & 0.03564026685975644 \\
140 & 0.03051237282492831 \\
160 & 0.026674947319897818 \\
180 & 0.023695169712517998 \\
200 & 0.021314357312846963 \\
220 & 0.019368374581756907 \\
240 & 0.017748046896721537 \\
260 & 0.01637792978607196 \\
280 & 0.015204215010418284 \\
300 & 0.014187492529606827
        \end{array}
      \end{displaymath}
\end{subfigure}%
    \end{figure}

Since $[z^n]\rgf{1} \sim 8^n \widetilde{C}_1 n^{-\nicefrac{3}{2}}$, the relative error $\delta([z^n]\rgf{1})$ is inevitably tending to $0$ as $n \to \infty$. Remarkably, the error $\delta([z^n]\rgf{1})$ converges much slower than one would expect. With $n = 300$ the error is just of order $10^{-2}$. We observed similar results in the relative errors for higher $n$, where the convergence rate is even slower than in the case of $[z^n]\rgf{1}$. 

Our Mathematica \textregistered~scripts and an implementation of the algorithm computing $\rgf{m}$ are available at~\cite{mb-haskell-implementation}.

\section{Conclusion}\label{sec:conclusion}
We presented several basis-independent results regarding asymptotic properties of combinatory logic. In particular, we generalised previously known results for $\S \K$-combinators from~\cite{Bendkowski2015,dgkrtz}, showing that they span to arbitrary Turing-complete combinator bases. Exploiting the results of~\cite{Bendkowski16}, we gave a systematic approach to finding better lower bounds on the density of normalising $\S \K$-combinators, improving the previously best known lower bound from about $3\%$ to approximately $34\%$. Performed super-computer experiments suggest that the searched density, if it exists, should be approximately $85\%$, conjecturing that both lower and upper bounds are still quite far from the actual density. Naturally, the following interesting question emerges -- is our approach converging to the actual density? As the asymptotic density is always bounded, the series $\sum_{m\geq 0} \density{R_m}{\class{\comb{S}\comb{K}}}$ converges to some value $\zeta \in (0,1)$. Is $\zeta$ the desired asymptotic density of normalising $\S \K$-combinators? If not, how far is it from the actual density? Due to the immense computational resources required in the computations and the sheer amount of major normal-order reduction grammar productions, our approach renders brute-force methods of closing the density gap virtually impossible. We expect that more sophisticated techniques are required in order to address this intriguing open problem.
\bibliographystyle{plain}
\bibliography{references}

\end{document}